\documentclass[a4paper]{article}
\usepackage[margin=3cm]{geometry}
\usepackage{latexsym}
\usepackage{amsmath}
\usepackage{amsthm}
\usepackage{amssymb}
\usepackage{amsfonts}
\usepackage[pdftex]{graphicx}

\newtheorem{theorem}{Theorem}
\newtheorem{remark}{Remark}
\newtheorem{definition}[theorem]{Definition}
\newtheorem{property}[theorem]{Proposition}
\newtheorem{notation}[theorem]{Notation}

\begin{document}

\title{Lyapunov exponents of a class of piecewise continuous systems of fractional order}

\author{MARIUS-F. DANCA\\Department of Mathematics and Computer Science, Avram Iancu University, \\Str. Ilie Macelaru, nr. 1A, 400380 Cluj-Napoca, Romania,\\and\\Romanian Institute of Science and Technology, \\Str. Ciresilor nr. 29, 400487 Cluj-Napoca, Romania}

\maketitle

\begin{abstract}
In this paper, we prove that a class of autonomous piecewise continuous systems of fractional order has well-defined Lyapunov exponents. For this purpose, based on some known results from differential inclusions of integer and fractional order and differential equations with discontinuous right-hand side, the associated discontinuous initial value problem is approximated with a continuous one of fractional order. Then, the Lyapunov exponents are numerically determined using, for example, the known Wolf's algorithm. Three examples of piecewise continuous chaotic systems of fractional order are simulated and analyzed: Sprott's system, Chen's system and Simizu-Morioka's system.
\end{abstract}


\section{Introduction}

\label{sec1}

Nowadays discontinuous systems of fractional order represent a novel topic of broad interest since they provide a logical link between the fractional derivative approach to descriptive systems and physical system properties, such as dry friction, forced vibration, brake processes with locking phases, as well as stick and slip phenomena.

However, to the best of our knowledge, there are yet very few works and results on discontinuous systems of fractional order. Also, most dedicated numerical methods for differential equations of fractional order can be used to "integrate" abruptly discontinuous equations of fractional order systems, but without mathematical justification (discontinuous equations may not even have classical solutions).

In this context, defining and calculating Lyapunov exponents (LEs) of systems modeled by fractional-order differential equations (FDEs) with discontinuous righthand side, represent a real challenge (see e.g. the \cite{kun} pp. 237--231, \cite{gran,gan} and \cite{li,wei,capo,ngu} and the references there, on calculating LEs in PWC and FDE systems respectively).

In this paper, the existence of LEs of piecewise continuous systems of fractional order is proved.

The systems are modeled by the following Caputo-type autonomous piece wise continuous (PWC) Initial Value Problem (IVP)
\begin{equation}\label{ivp}
D_*^qx=f(x):=g(x)+A(x)s(x),~~~x(0)=x_0,~~~ t\in[0,T],
\end{equation}

\noindent where $T>0$, $g:\mathbb{R}^n\rightarrow \mathbb{R}^n$ is a nonlinear, at least continuous, function, $s:\mathbb{R}^n\rightarrow \mathbb{R}^n$, $s(x)=(s_1(x_1),s_2(x_2),...,s_n(x_n))^T$ a piece-wise function, with $s_i:\mathbb{R}\rightarrow \mathbb{R}$, $i=1,2,...,n$, piece-wise constant functions ($sign$ or Heaviside functions in many applications), $A\in \mathbb{R}^{n\times n}$ a square matrix of real functions, and $D_*^q$, with $q$ being some positive real number, stands for the Caputo fractional derivative.

The discontinuity of $f$ is assured if the following assumption is considered:

\vspace{3mm}
\noindent \textbf{(H1)} At least one element of $A(x)s(x)$ is discontinuous.
\vspace{3mm}

As in most of practical examples, $s_i(x_i)=sign(x_i)$. For example, for the usual case $n=3$ and $s(x)=(sign(x_1),sign(x_2),sign(x_3))^T,~$ let us consider the fractional-order variant of PWC Sprott's system \cite{sprotus,sprot2}
\begin{equation}\label{spr}
\begin{array}{l}
D_{\ast }^{q_1}x_{1}=x_{2}, \\
D_{\ast }^{q_2}x_{2}=x_{3}, \\
D_{\ast }^{q_3}x_{3}=-x_{1}-x_{2}-ax_{3}+bsign(x_{1}),%
\end{array}%
\end{equation}

\noindent with $a=0.5$, $b=1$, where
\[
g(x)=\left(
\begin{array}{c}
x_{2} \\
x_{3} \\
-x_{1}-x_{2}-ax_{3}%
\end{array}%
\right) ,~~~A=\left(
\begin{array}{ccc}
0 & 0 & 0 \\
0 & 0 & 0 \\
0 & 0 & b%
\end{array}%
\right).
\]

Replacing the PWC functions, such as $sign$ or Heaviside function, with continuous functions, represents an usual setting in many works (see e.g. \cite{wir}). Using some known results of differential inclusions, following the way presented in \cite{danda} we show why and how this approximation can be done.

In this paper we prove that systems modeled by the IVP (\ref{ivp}) can be approximated with continuous systems of fractional order, for which the variational equations defining LEs are well defined.

The paper is organized as follows: Section I presents the notions and results utilized in this paper, Section II presents the way in which the IVP (\ref{ivp}) can be continuously approximated and Section III deals with the variational equations which define LEs. In Section IV the LEs for three examples of PWC systems of fractional order are determined. The Conclusion Section ends this paper.

\section{Preliminaries}

\begin{notation}
Denote by $\mathcal{M}$ the null \emph{discontinuity set} of $f$, generated by the discontinuity points of $s_i$.
\end{notation}
\noindent $\mathcal{M}$ has zero Lebesgue measure, $\mu(\mathcal{M})=0$, and divides $\mathbb{R}^n$ to several $m>1$ open disjunct and connected sub-domains $\mathcal{D}_i\subset \mathbb{R}^n$, $i=1,2,...,m$, such that $\mathbb{R}^n=\bigcup_{i=1}^m \overline{\mathcal{D}_i}$. The discontinuity points belong to the union of the boundaries of $\mathcal{D}_i$, i.e. $\mathcal{M}=\bigcup_{i=1}^m bndry(\mathcal{D}_i)$ (see e.g. \cite{cort} for a tutorial on discontinuous dynamical systems).

\begin{definition} A function $f : \mathbb{R}^n\rightarrow \mathbb{R}^n$, is called\emph{ piece-wise continuous } if it is continuous throughout $R^n\setminus  \mathcal{M}$ and at $\mathcal{M}$ has finite (possible different) limits.
\end{definition}
Under Assumption \textbf{H1}, $f$, defined by (\ref{ivp}), is PWC.

\noindent The following assumption on $g$ will be considered:

\vspace{3mm}
\noindent (\textbf{H2}) $g$ is differentiable on $\mathbb{R}^n$.
\vspace{3mm}

Because the PWC functions, $s_i$, are linear on $\mathcal{D}_k$, $k=1,2,...,m$, they are differentiable on $\mathcal{D}_k$. Therefore, the following property holds

\begin{property}
$f$ is PWC on $\mathbb{R}^n$ and differentiable on $\mathcal{D}_k$, $k=1,2,...,m$.
\end{property}

\noindent The differentiability of $g$ is required for LE. \footnote{Actually, in the great majority of known examples, with $g$ being polynomial, it is also a smooth function.}

\noindent For example, for the PWL function $f:\mathbb{R}\rightarrow \mathbb{R}$ defined by
\begin{equation}\label{exemplu}
f(x)=2-3sgn(x),
\end{equation}

\noindent the set $\mathcal{M}=\{0\}$ determines the continuity (and also differentiability) sub-domains $\mathcal{D}_1=(-\infty, 0)$, $\mathcal{D}_2=(0,\infty)$ (see the graph in Fig. \ref{fig2} a).

\begin{definition}
Let $x:[0,T]\rightarrow \mathbb{R}$ and $q>0$. The \emph{Caputo fractional derivative with starting point }$0$, introduced by M. Caputo in 1967 in \cite{cap}, is defined as
\begin{equation}\label{cap}
D_*^qx(t)=\frac{1}{\Gamma(n-q)}\int_0^t (t-\tau)^{n-q-1}x^{(n)}(\tau)d\tau.
\end{equation}
\end{definition}

\noindent The fractional order, $q\in(n-1,n)$ with $n\in \mathbb{Z}^+$, being the nearest integer bigger than $q$ ($n=\lceil q\rceil$). $\Gamma$ is Euler's Gamma function, a generalization of the factorial function $n!$, i.e. $\Gamma(n)=(n-1)!$, $n\in \mathbb{N}$, defined as

\begin{equation*}
\Gamma(z)=\int_0^t t^{z-1}e^{-t}dt, ~~~ z\in\mathbb{C}, ~~Re(z)>0.
\end{equation*}

Compared to other fractional-order differential operators, $D_*^q$ is more restrictive since it requires the $n$-th (\emph{first}, for $q<1$) derivative of $f$ (see e.g. \cite{old,pod,nico}). However, the Caputo derivative with starting point $0$ has the great advantage that it uses the fractional order initial conditions. Therefore, the use of Caputo's derivative in the IVP (\ref{ivp}) is fully justified because in practical (physical) problems, we need physically interpretable initial conditions (see e.g. \cite{nico,pod5}). Thus, in (\ref{ivp}), the initial condition(s) can be used as the integer-order differential equation counterpart, which, for the common case of $q\in(0,1)$, will reduce to $x(0)=x_0$.

If we consider the FDE associated with (\ref{exemplu})
\begin{equation}\label{exemplu2}
D_*^qx=2-3sign(x),~~~x(0)=x_0,
\end{equation}

\noindent then there are no classical (continuously differentiable) solutions starting from some point $x_0$. Thus, for $x=x_0 = 0$, there is no solution ($D_*^q(0)=0\neq2=2-3sign(0)$). For $x_0 > 0$, there exists a solution but only on the interval $[0, T')$ with $T' = (\Gamma(1+q) x_0 )^{1/q}$. This solution has the form $x(t) = x_0 -t^q/\Gamma(1+q)$, and it cannot be extended to any interval larger than $[0, T')$. For $x_0 < 0$, there also exists some $T'' > 0$, $T''={(\Gamma(1+q)x_0/5)}^{1/q}$,  such that the solution, $x(t) = x_0 +5 t^q/\Gamma(1+q)$, exists but only on $[0, T'')$. Even these solutions tend to the line $x=0$, they cannot extend along this line (see Fig. \ref{fig3} a, where $q=0.6$ and $q=0.8$).

One way to overcome this difficulty is similar to the one for integer DEs with discontinuous right-hand side (Fillipo equations), namely to transform the discontinuous right-hand side into a convex set-valued function with closed values. In this way, the problem is restarted as a differential inclusion of fractional order.

\begin{definition} A \emph{set-valued} (\emph{multi-valued}) function $F:\mathbb{R}^n\rightrightarrows \mathbb{R}^n$ is a function which associates to any element $x\in\mathbb{R}^n$, a subset of $\mathbb{R}^n$, $F(x)$ (the image of $x$).
\end{definition}

There are several ways to define $F(x)$. The (convex) definition was introduced by Filippov in \cite{filip} (see also\cite{dai,aub1,aub2})
\begin{equation}\label{fill}
F(x)=\bigcap_{\varepsilon >0}\bigcap_{\mu(\mathcal{M})=0} \overline{conv}(f({z\in \mathbb{R}^n: |z-x|\leq\varepsilon}\backslash \mathcal{M})).
\end{equation}
Here, $F(x)$ is the closure of the convex hull of $f(x)$, with $\varepsilon$ being the radius of the ball centered at $x$. At any continuity point of $f$, $F(x)$ consists of one single point, which coincides with the value of $f$ at this point (i.e. we get back $f(x)$ as the right-hand side: $F(x)=\{f(x)\}$), while at the points belonging to $\mathcal{M}$, $F(x)$ is given by (\ref{fill}).

If $s_i$ are $sign$ functions, the underling set-valued form, denoted by $Sign:\mathbb{R}\rightrightarrows \mathbb{R}$, is defined as follows
\begin{equation}
Sign(x)=\left\{
\begin{array}{cc}
\{-1\}, & x<0, \\
\lbrack -1,1], & x=0, \\
\{+1\}, & x>0.%
\end{array}%
\right.
\end{equation}

Thus, $sign(0)$ is taken now as the whole interval $[-1 , 1]$ ''connecting'' the points $-1$ and $+1$.

In Fig. \ref{fig4} a and Fig. \ref{fig4} b there are plotted $sign$ and $Sign$, respectively, and in Fig. \ref{fig3} b the graph of the set-valued function corresponding to Example (\ref{exemplu}).

\noindent In this way, a discontinuous FDE of the form
\begin{equation}\label{disc}
D_*^qx=f(x),~~~ x(0)=x_0,~~~ t\in [0,T],
\end{equation}

\noindent can be restarted as a set-valued problem (fractional-order differential inclusion (FDI))
\begin{equation}\label{set_v}
D_*^qx\in F(x),~~~ x(0)=x_0, ~~~\text{for a.a.}~~~ t\in [0,T].
\end{equation}

Differential inclusions of integer-order have been intensively studied in the literature since the 1930s, one of the first works being attributed to the Polish mathematician Zaremba \cite{zar} (see also \cite{filip,dai,aub1,aub2}), while the study of fractional order differential inclusions was initiated by El-Sayed and Ibrahim in \cite{el}. Some of the (few) works on differential inclusions of fractional order are the papers \cite{jon,jon2,yong}.

Following the way proposed by Filippov for DE of integer-order with discontinuous right-hand side, a (generalized or Filippov) solution to (\ref{disc}) can be defined as follows (\cite{filip} p.85; see also \cite{waz}).

\begin{definition}  A \emph{generalized solution} to (\ref{disc}) is an absolutely continuous function $x:[0,T]\rightarrow\mathbb{R}$ satisfying (\ref{set_v}) for a.a. $t\in[0,T]$.
\end{definition}

Thus, by applying Filippov's regularization, the problem (\ref{exemplu}) is converted to the following set-valued IVP of fractional order
\begin{equation}\label{ex_set}
D_{\ast }^{q}(x)-2\in-3Sign(x)=\left\{
\begin{array}{lc}
5, & x<0, \\
\left[ -1,5\right] , & x=0,\\
-1, & x>0,%
\end{array}%
\right.
\end{equation}

\noindent for a.a. $t\in [0,T]$. Now, for $x_0=0$, one can choose for $D_*^q$ any value in $[-1,5]$, for example $0$ (Fig. \ref{fig2} b), and the equation at this point reads $D_*^qx=0$. This shows that $x(t)=0$, for $x_0=0$, is a solution which verifies now the equation. In this way, using (\ref{ex_set}), the FDE (\ref{exemplu2}) will have a generalized solution. Therefore, for whatever initial condition $x_0$, the solution will reach, and also continue forward, the line $x=0$ (dotted line in Fig. \ref{fig3} b).

Since the theory of numerical methods for FDI is only at the begining (\cite{rob} is one of the very few works on this subject), we propose a different approach to integrate numerically the underlying set-valued IVP: to approximate the set-valued problem with a single-value continuous one of fractional order, for which there are known numerical methods, such as the predictor-corrector Adams-Bashforth-Moulton (ABM) method \cite{kai1}.

\begin{definition} Let $F:\mathbb{R}^n\rightrightarrows \mathbb{R}^n$ be a set-valued function; the single-valued function $f:\mathbb{R}^n\rightarrow \mathbb{R}^n$ is called a \emph{selection} (\emph{approximation}) of $F$ if $f(x)\in F(x)$ for every $x\in \mathbb{R}^n$ (see e.g. \cite{cll}, \cite{kas}).
\end{definition}

\begin{definition} The map $F$ is \emph{upper semi-continuous} (USC) on $\mathbb{R}^n$ if for each $x_0\in\mathbb{R}^n$, the set $F(x_0)$ is a nonempty and closed subset of $\mathbb{R}^n$, and if for each open set $N$ of $\mathbb{R}^n$ containing $F(x_0)$, there exists an open neighborhood $M$ of $x_0$ such that $F(M)\subset N$.
\end{definition}

\begin{theorem}\emph{(\textbf{Cellina's Theorem} \label {th0}\cite{aub1,aub2})}. Let $F:\mathbb{R}^n\rightrightarrows\mathbb{R}^n$ be USC. If the values of $F$ are nonempty and convex, then for every $\varepsilon>0$, there exists a locally Lipschitz selection $f_{\varepsilon}:\mathbb{R}^n\rightarrow \mathbb{R}^n$ such that
\begin{equation*}
Graph(f_\varepsilon)\subset Graph(F)+ \varepsilon B,
\end{equation*}

\end{theorem}

\begin{remark}
Due to the symmetric interpretation of a set-valued map as a graph (see e.g. \cite{aub1}), we shall say that a set-valued map satisfies a property if and only if its graph satisfies it. For instance, a set-valued map is said to be convex if and only if its graph is a convex set.
\end{remark}

For the set-valued function $F(x)=2-3Sign(x)$, a sketch of a selection is plotted in Fig. \ref{fig2} c.

\section{Continuous approximation of $f$ }\label{switch}

\noindent Applying the Filippov regularization to (\ref{ivp}) leads to the following FDI
\begin{equation}
\label{IVP1}
D_*^qx\in F(x):=g(x)+A(x)S(x),
\end{equation}

\noindent with
\begin{equation}\label{s}
S(x)=(S_1(x_1),S_2(x_2),...,S_n(x_n))^T,
\end{equation}

\noindent where $S_i:\mathbb{R}\rightarrow \mathbb{R}$ are the set-valued variants of $s_i$, $i=1,2,...,n$ ($Sgn(x_i)$ for the usual case of $sgn(x_i)$).

The existence of solutions for FDIs is analyzed in, e.g., \cite{el}. However, because of the lack of numerical methods to find the solutions to fractional-order differential inclusions, necessary to determine LEs, we can convert the set-valued IVP to a single-valued one of fractional order.

It is easy to check that $F$, defined in (\ref{IVP1}), verifies the requirements of Theorem \ref{th0} (see \cite{dand}). Therefore, we can enounce the following theorem:
\begin{theorem}\label{th1}\emph{\cite{danda}}
The PWC IVP of fractional order (\ref{ivp}) can be transformed into the following continuous IVP of fractional order,
\begin{equation}\label{glo}
D_*^qx={f}_\varepsilon(x):=g(x)+A(x){s}_\varepsilon(x),
\end{equation}

\noindent where ${s}_{\varepsilon}(x)=({s}_{1\varepsilon}(x_1),{s}_{2\varepsilon}(x_2),...,{s}_{n\varepsilon}(x_n))^T$ is the continuous approximation of $S(x)=(S_1(x_1),S_2(x_2),...,S_n(x_n))^T$.
\end{theorem}

The steps of the proof can be viewed in the sketch drawn in Fig. \ref{fig5}. As can be seen, the constructive proof allows to choose the approximation simply by replacing $s$ with the continuous function $s_\varepsilon$ (grey line in Fig. \ref{fig5}).

The existence of (Lipschitz) continuous approximations for the set-valued function $F$ defined in (\ref{IVP1}) is studied in \cite{danx}.

The approximation of $f$ defined in (\ref{ivp}) can be done as closely as desired, and can be made locally (in small neighborhoods of discontinuity points of $s$), or globally (in small neighborhoods of the graph of $S$; see Fig. \ref{fig2} c) \cite{danda}.

\noindent Generally, a set-valued function admits (infinitely) many local or global approximations.

In this paper, we use global approximations, which are easy to implement numerically.

For the sake of simplicity, for each component $s_{i \varepsilon}(x_i)$, $i=1,2,...,n$, $\varepsilon$ can have the same value.

Since most of practical examples of PWC systems are modeled via $sign$ function, we shall use for its global approximation, the so-called \emph{sigmoid functions}, ${sign_\varepsilon}$, because this class of functions provide the required flexibility and to which the abruptness of the discontinuity can be easily modified\footnote{The class of sigmoid functions includes many other examples such as the hyperbolic tangent, the error function, the logistic function, algebraic functions like $\frac{x}{\sqrt{\delta+x^2}}$, $\frac{2}{1+e^{-\frac{x}{\delta}}}-1$ \cite{danda} and so on.}
\begin{equation}\label{h_simplu}
{sign}_{\varepsilon(\delta)}(x)=\frac{2}{\pi}arctan\frac{x}{\delta}\approx Sign(x).
\end{equation}

\noindent In (\ref{h_simplu}), $\delta$ is a positive parameter which controls the slope in the $\varepsilon$-neighborhood of the discontinuity $x=0$ (In Fig. \ref{fig6} a the graph is plotted for a large value of $\delta$, $\delta=1E-1$, for a clear image).

For global approximations, $\varepsilon$ is determined implicitly, depending proportionally on $\delta$ size, $\varepsilon=\varepsilon(\delta)$. For example, in order to obtain a neighborhood $\varepsilon$ of order of $1E-6$, we need to choose for $\delta$ the value $1E-4$ (in \cite{danda}, a detailed numerical analysis for the case of the sigmoid function $sign_\varepsilon(x)=2/(1+e^{-x/\delta})-1$ is presented).

The function in example (\ref{exemplu}) becomes
\begin{equation}
f(x)\approx f_{\varepsilon(\delta)}(x)=2-3{sgn}_{\varepsilon(\delta)}(x)=2-\frac{6}{\pi}arctan\frac{x}{\delta}.
\end{equation}

\section{Lyapunov exponents}

As for the case of integer-order systems, in order to determine the Lyapunov spectrum, we show next that it is possible to find a differential (variational) equation of fractional order.

Let us consider again the IVP (\ref{ivp}). The following theorem holds

\begin{theorem} System (\ref{ivp}) has the following variational equations which define the LEs
\begin{equation}\label{vari}
\begin{array}{l}
D_{\ast }^{q}\Phi (t)=D_{x}f_\varepsilon(x)\Phi (t), \\
\Phi (0)=I_{n}.%
\end{array}%
\end{equation}

\end{theorem}

\begin{proof}

By Theorem \ref{th1}, $f$ is approximated by $f_\varepsilon$. Under the considered assumptions, there exists a flow \cite{kai0} $\phi:I\times \mathbb{R}^n\rightarrow \mathbb{R}^n$ of the approximated system system (\ref{glo}), $x=\phi(t,x)$, which satisfy $ D_*^q \phi(t,x_0)=f_\varepsilon(\phi(t,x_0))$, $\phi(0,x_0)=x_0$, for all $t\in I$. Next, Theorem 2 in \cite{li} ensures the existence of the variational equation (\ref{vari}), where $D_{x}f_\varepsilon(x)$ is the matrix of the partial derivatives of $f_\varepsilon$ evaluated along the solution $x=\phi(t,x)$\footnote{Because while the problem is solved in parallel with (\ref{glo}) the initial conditions change, the usual notation $x_0$ is replaced with $x$.}, $\Phi(t)=D_x\phi(t,x)$ is the Jacobian of the flow $\phi$ evaluated at $\phi(t,x)$ and $I_n$ is the identity matrix.
\end{proof}

The parallel solving of systems (\ref{glo}) and (\ref{vari}), which is a linear matrix-valued and time-varying system with coefficients depending on the evolution of the original system (\ref{IVP1}), allows us to determine LEs.

\begin{remark}
This result can be extended for Cauchy problems involving Riemann-Liouville derivative \cite{li}. However, due to the benefit of using initial conditions as for integer order, we restrict the study to Caputo's derivative.
\end{remark}

Oseledec's Multiplicative Ergodic Theorem \cite{ose}, which ensures the existence of LEs from the stability matrix of the system, still applies to fractional-oredr systems, and proves that the entire spectrum of LEs can be determined and the following limit exists and is finite
\begin{equation}\label{lamb}
\lambda_k=\lim _{t\rightarrow \infty}\frac{1}{t}ln|\Lambda_k(t)|, ~~~k=1,2,...,n.
\end{equation}

\noindent In (\ref{lamb}), $\Lambda_k(t)$ are the eigenvalues of $\Phi(t)$ and, following the common order convention, the ordered LEs: $\lambda_1\geq \lambda_2\geq \cdots\geq \lambda_n$ measure the rate of growth of infinitesimal $k$-dimensional volumes, $k=1,...,n$.

Because the analytic evaluation is rarely available, usually the LEs are computationally estimated.

Therefore, by replacing in the known algorithms for LEs spectrum or for the largest LE, the numerical method for ODEs of integer order (usually the 4th RK method) with some methods for FDEs (ABM method in this paper), one can determine the LEs spectrum, the largest LE or their bounds (see e.g. \cite{li} or \cite{wei} for LEs spectrum and largest LE).

\section{Applications}

In this section, the LE spectrum is determined with the Matlab code \emph{lyapunov.m} \cite{lyap}, which adapts Wolf's algorithm  \cite{wol} for continuous systems of integer order. Thus, in Wolf's algorithm, the ODEs integrator for continuous flows (usually the classical Runge–Kutta fourth-order method) is replaced with some numerical schemes for FDEs.\footnote{ In this paper, the Matlab code \emph{fde12.m} \cite{code} (an implementation of the predictor-corrector PECE method of Adams-Bashforth-Moulton type presented in \cite{kai1}, which has been suitably modified for the general case with the incommensurate case) has been utilized.} The integration step size is $h=0.001$ and the dynamics of the LE spectrum have been determined for $t\in[0,400]$, while the chaotic attractors have been drawn for $t\in[0,1000]$.

Let us consider the usual case of the PWC $sign$ function. In order to integrate the variational equation (\ref{vari}), we need to calculate the derivative of the sigmoid function $sign_\varepsilon$
\begin{equation}\label{derivata}
sign'_{\varepsilon(\delta)}(x)=\frac{d}{dx}sign_{\varepsilon(\delta)}(x)=\frac{2}{\pi}\frac{\delta}{\delta^2+x^2}.
\end{equation}

At $x=0$, due to the vertical segment, the set-valued function $Sign$ has the slope $+\infty$, while any approximation $f_{\varepsilon(\delta)}$ of $Sign$ has a finite slope given by the derivative at $x=0$, as

\begin{equation*}
m=sign'_{\varepsilon(\delta)}(0)=\frac{2}{\pi\delta}.
\end{equation*}

The graph of $sign_{\varepsilon(\delta)}$ and $sign'_{\varepsilon(\delta)}$, chosen for clarity $\delta=1E-1$, are plotted in Fig. \ref{fig5}.

In this paper, we consider $\delta=5E-4$. For this value, $\varepsilon$ was of order of $1E-6$, and $m=4000/\pi$. Lower values for $\delta$ imply higher values for the derivative of $sign_\varepsilon$ at $x=0$. For example, for $\delta=1E-5$, $m=2E5/pi$ which, considering the roundoff errors and also the errors of ABM method \cite{kai1,kai2}, it can finallz lead to a loss of precision.

\begin{itemize}
\item
The Sprott system (\ref{spr}), for $q=(0.99,0.98,0.97)$ and the chosen parameters $a=1$ and $b=0.5$, behaves chaotically (Fig. \ref{fig7}) and has the approximated form
\begin{equation}
\begin{array}{l}
D_{\ast }^{0.99}x_{1}=x_{2}, \\
D_{\ast }^{0.98}x_{2}=x_{3}, \\
D_{\ast }^{0.97}x_{3}=-x_{1}-x_{2}-0.5x_{3}+sign_{\varepsilon(\delta)}(x_{1}),%
\end{array}%
\end{equation}

\noindent with $sign_{\varepsilon(\delta)}(x_1)$ given by (\ref{h_simplu}), and the Jacobian
\[
D_{x}f_{\varepsilon(\delta)}(x)=\left(
\begin{array}{ccc}
0 & 1 & 0 \\
0 & 0 & 1 \\
-1+sign'_{\varepsilon(\delta)}(x_1) & -1 & -0.5%
\end{array}%
\right),
\]

\noindent where $sign'_{\varepsilon(\delta)}$ is given by (\ref{derivata}).

The LE spectrum, obtained after running the code on $I=[0,400]$ (see the dynamics of LEs in Fig.\ref{fig8}), is $\{ 0.115,  -0.006,  -0.730\}$.

\item Let us next consider the fractional+order variant of the piece-wise linear (PWL) Chen's system \cite{aziz}
\begin{equation}
\begin{array}{l}\label{ch}
D_{\ast }^{q_1}x_{1}=a\left( x_{2}-x_{1}\right),  \\
D_{\ast }^{q_2}x_{2}=sign(x_{1})\left(c-a-x_{3}\right) +cdx_{2}, \\
D_{\ast }^{q_3}x_{3}=x_1sign(x_{2})-bx_{3},
\end{array}
\end{equation}

\noindent with coefficients $ a = 1.18,~ b = 0.16, ~c = 1.2, ~d = 0.1$, where

\[
g(x)=\left(
\begin{array}{c}
a(x_2-x_1) \\
cdx_2 \\
-bx_3%
\end{array}%
\right) ,~~~A(x)=\left(
\begin{array}{ccc}
0 & 0 & 0 \\
c-a-x_3 & 0 & 0 \\
0 & x_1 & 0%
\end{array}%
\right).
\]

For $q=(0.99,0.9,0.999)$, and for the considered coefficients values, the system behaves chaotically (Fig.\ref{fig9}), and the approximated form is
\begin{equation}
\begin{array}{l}
D_{\ast }^{0.99}x_{1}=1.18\left( x_{2}-x_{1}\right),  \\
D_{\ast }^{0.9}x_{2}=sign_{\varepsilon(\delta)}(x_{1})(0.82-x_{3}) +0.12x_{2}, \\
D_{\ast }^{0.999}x_{3}=x_1sign_{\varepsilon(\delta)}(x_{2})-0.16x_{3},
\end{array}
\end{equation}

\noindent which has the Jacobian

\[
\begin{array}{l}
D_{x}f_{\varepsilon (\delta )}(x)= \\\\
\left(
\begin{array}{ccc}
-1.18 & 1.18 & 0 \\
0.82sign_{\varepsilon (\delta )}^{\prime }(x_{1}) & 0.12 &
-sign_{\varepsilon (\delta )}(x_{1}) \\
sign_{\varepsilon (\delta )}^{{}}(x_{2}) & x_{1}sign_{\varepsilon (\delta
)}^{\prime }(x_{2}) & -0.16%
\end{array}%
\right).
\end{array}%
\]

\noindent The LEs are: $\{0.719, 0.289, -2.191\}$ (see Fig. \ref{fig10}).

Because there are two positive exponents, the PWC Chen's system of fractional order is hyperchaotic.

\item The last considered example is the fractional+oredr variant of the PWC Shimizu--Morioka's system \cite{shimi,shimi1} with $q=(0.99,0.97,0.98)$
\begin{equation}\label{simiz}
\begin{array}{l}
D_{\ast }^{0.99}x_{1}=x_{2}, \\
D_{\ast }^{0.97} x_{2}=(1-x_3)sign(x_{1})-a x_2 , \\
D_{\ast }^{0.98} x_{3}=x_1^2-b x_3,%
\end{array}%
\end{equation}

\noindent which with $a=0.75$ and $b=0.45$ behaves chaotically (Fig. \ref{fig11}). Here $g(x)=(x_2,-0.75x_2,x_1^2-0.45x_3)^T$ and

\[
A(x)=\left(
\begin{array}{ccc}
0 & 0 & 0 \\
1-x_{3} & 0 & 0 \\
0 & 0 & 0%
\end{array}%
\right).
\]
The approximated system is
\begin{equation}\label{simiz2}
\begin{array}{l}
D_{\ast }^{0.99}x_{1}=x_{2}, \\
D_{\ast }^{0.97} x_{2}=(1-x_3)sign_{\varepsilon(\delta)}(x_1)-0.75 x_2 , \\
D_{\ast }^{0.98} x_{3}=x_1^2-0.45 x_3,%
\end{array}%
\end{equation}

\noindent and the Jacobian has the following form
\[
\begin{array}{l}
D_{x}f_{\varepsilon (\delta )}(x)= \\ \\
\left(
\begin{array}{ccc}
0 & 1 & 0 \\
(1-x_3)sign'_{\varepsilon(\delta)}(x_1) & -0.75 & -sign_{\varepsilon(\delta)}(x_1)\\%
-sign_{\varepsilon (\delta )}(x_{1}) \\
2x_1 & 0 & -0.45\\
\end{array}%
\right).
\end{array}%
\]

\noindent The LEs are: $\{0.236, 0.137,  -1.828\}$ (see Fig. \ref{fig12}).

Similar to Chen's system, Shimizu--Morioka's system is hyperchaotic due to the presence of two positive LEs.

\end{itemize}

\section{Conclusion}
In this paper, we have shown that PWC systems of fractional order have well-defined LEs. To prove that the PWC systems can be continuously approximated, Cellina's Theorem and some results from the theory of differential equations with discontinuous right-hand side have been utilized.

The approximation of the discontinuous elements on the right-hand side of the IVPs, which generally are $sign$ functions, has been realized with the sigmoid function $2/\pi atan(x/\delta)$, with $\delta$ being a parameter which determines the slope of approximation in small neighborhoods of the discontinuity point $x=0$.

The variational equations which define the LEs help to find numerically the LEs. For this purpose, we used Matlab implementations of the known Wolf algorithm in which the numerical method for ODEs has been replaced with the predictor-corrector ABM method for fractional-order differential equations.

A future task related to this subject would be the numerical analysis of the computational errors given by any algorithm for numerical determination of LEs and also given by the numerical methods for FDEs. As is well known, the errors in these cases are quite large and, therefore, an optimal integration step size correlated to the maximal integration interval would be of importance.

\newpage

\begin{figure*}
\begin{center}
\includegraphics[scale=0.75] {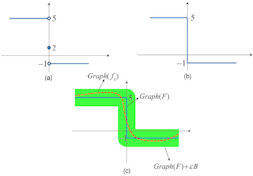}
\caption{Graph of the function $f(x)=2-3sign(x)$. a) Discontinuity at $x=0$. b) Graph of the set-valued version of $f$. c) Sketch of a continuous selection of the set-valued function $F$ within an $\varepsilon$-neighborhood of $F$.}
\label{fig2}       
\end{center}
\end{figure*}

\begin{figure*}
\begin{center}
\includegraphics[scale=0.7] {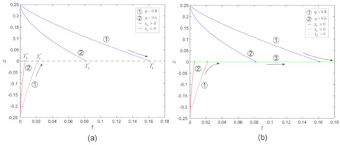}
\caption{Solutions of the equation $D_*^qx=2-3sign(x)$ for $q=0.8$ and $q=0.6$. a) Depending on the sign of $x_0$, the solutions tend to the line $x=0$ but they cannot reach this line; there is no a classical solution. b) Considering the differential inclusion $D_*^qx-2\in-3Sign(x)$, there exists a generalized solution which continues through the line $x=0$ for whatever $x_0$.}
\label{fig3}       
\end{center}
\end{figure*}

\begin{figure*}
\begin{center}
\includegraphics[scale=0.75] {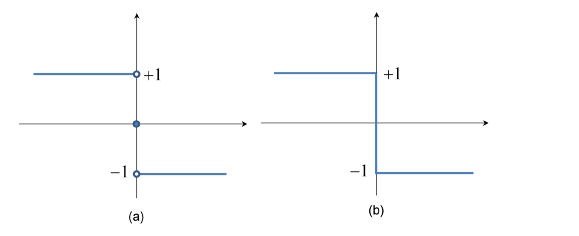}
\caption{a) Graph of $sign(x)$. b) Graph of $Sign(x)$.}
\label{fig4}       
\end{center}
\end{figure*}

\begin{figure*}
\begin{center}
\includegraphics[scale=0.5] {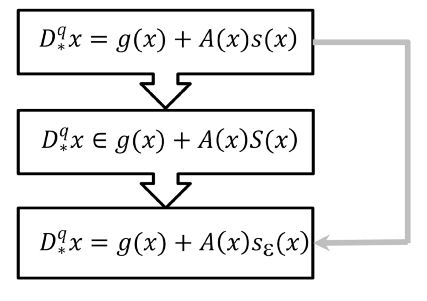}
\caption{The sketch of the proof steps of Theorem \ref{th1}.}
\label{fig5}       
\end{center}
\end{figure*}

\begin{figure*}
\begin{center}
\includegraphics[scale=0.75] {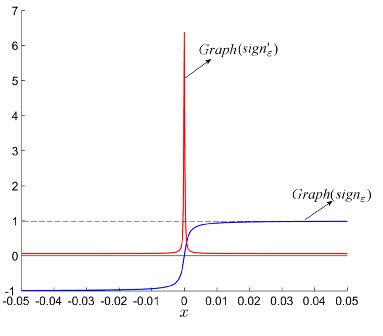}
\caption{Graphs of $sign_{\varepsilon(1e-1)}$ (blue plot) and $\frac{d}{dx}sign_{\varepsilon(1e-1)}$ (red plot).}
\label{fig6}       
\end{center}
\end{figure*}

\begin{figure*}
\begin{center}
\includegraphics[scale=0.75] {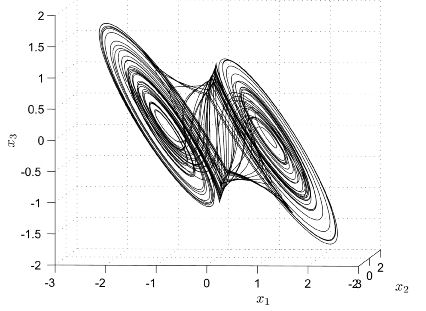}
\caption{Sprott's chaotic attractor.}
\label{fig7}       
\end{center}
\end{figure*}

\begin{figure*}
\begin{center}
\includegraphics[scale=0.75] {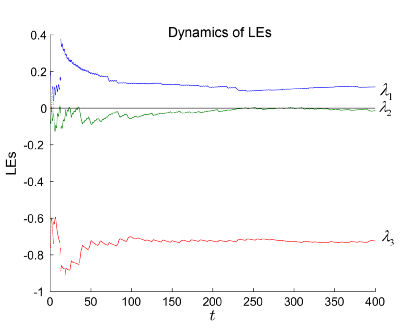}
\caption{Dynamics of the LEs exponents for the Sprott system.}
\label{fig8}       
\end{center}
\end{figure*}

\begin{figure*}
\begin{center}
\includegraphics[scale=0.75] {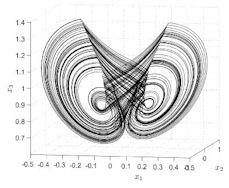}
\caption{Chen's chaotic attractor.}
\label{fig9}       
\end{center}
\end{figure*}

\begin{figure*}
\begin{center}
\includegraphics[scale=0.75] {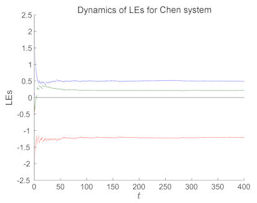}
\caption{Dynamics of the LEs exponents for the Chen system.}
\label{fig10}       
\end{center}
\end{figure*}

\begin{figure*}
\begin{center}
\includegraphics[scale=0.75] {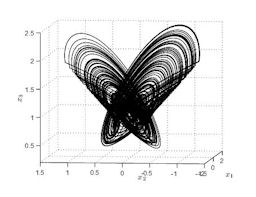}
\caption{Shimizu--Morioka's chaotic attractor.}
\label{fig11}       
\end{center}
\end{figure*}

\begin{figure*}
\begin{center}
\includegraphics[scale=0.75] {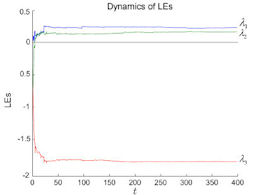}
\caption{Dynamics of the LEs exponents for the Shimizu--Morioka system.}
\label{fig12}       
\end{center}
\end{figure*}

\end{document}